\documentclass[conference,letterpaper]{IEEEtran}
\usepackage{booktabs}
\usepackage{amsmath}
\usepackage{amssymb}
\usepackage{caption}
\usepackage{algorithmic}
\usepackage{todonotes}
\usepackage{amsfonts}
\usepackage{cite}
\usepackage{booktabs}
\usepackage{graphicx}
\usepackage{array}
\usepackage[shortlabels]{enumitem}

\newtheorem{definition}{Definition}[section]
\newtheorem{theorem}[definition]{Theorem}

\newtheorem{lemma}[definition]{Lemma}

\ifCLASSINFOpdf
\else
\fi
\IEEEoverridecommandlockouts

\begin{document}

\title{Chordal Codes: Code Construction and Efficient Gaussian-Elimination Decoding}
\author{ 
   % \IEEEauthorblockN{Yulin Chen}
   % \IEEEauthorblockA{Department of Information Engineering \\
   %                   The Chinese University of Hong Kong\\
   %                   % Hong Kong\\
   %                   Email: cy022@ie.cuhk.edu.hk}
   % \and
   % \IEEEauthorblockN{Laigang Guo }
   % \IEEEauthorblockA{University}

   % \and
   % \IEEEauthorblockN{ Tong Li}
   % \IEEEauthorblockA{University}

   % \and
   % \IEEEauthorblockN{Mingyang Zhu }
   % \IEEEauthorblockA{University}
Yulin CHEN,\quad

Laigang GUO,\quad

Tong LI,\quad

Mingyang ZHU
\thanks{Yulin Chen is with the Department
of Information Engineering, The Chinese University of Hong Kong, Hong Kong, SAR, China e-mail: (1155189560@link.cuhk.edu.hk)(\emph{Corresponding authors: Mingyang Zhu; Laigang Guo})}
\thanks{Laigang Guo is with the Laboratory of Mathematics and Complex Systems (Ministry of Education), School of Mathematical Sciences, Beijing Normal University, Beijing 100875, China (e-mail: lgguo@bnu.edu.cn)}
\thanks{Tong LI is with the Academy of Mathematics and Systems Science, Chinese Academy of Sciences, Beijing, 100190, P. R. China e-mail:litong@amss.ac.cn}
\thanks{M. Zhu is with the School of Information and Electronic Engineering, East China Normal University, Shanghai 200241, China (e-mail: myzhu@ee.ecnu.edu.cn)}}

\maketitle

\begin{abstract}
    Maximum-likelihood decoding of linear codes over erasure channels can be considered as solving a linear system of a column submatrix of the parity-check matrix. In practice, Gaussian elimination is usually applied, and its computational cost depends on the structure of the matrix, especially the sparsity and the fill-in generated during the elimination. The \emph{chordal graph}, or \emph{triangulated graph}, first introduced by C. Berge, has a so-called \emph{perfect elimination ordering}, indicating that if the primal graph of a matrix is a chordal graph, then there exists a perfect elimination ordering such that a special class of fill-in under Gaussian elimination following this ordering will be avoided. In this paper, we introduce \emph{chordal codes}, i.e., codes equipped with parity-check matrices such that the primal graphs of the parity-check matrices are chordal. We prove that, for a chordal code, the decoding submatrix associated with any erasure pattern retains a chordal primal graph. So, the complexity of chordal elimination (Gaussian elimination following the perfect elimination ordering applied to the chordal code decoding) will be smaller than decoding a random code by Gaussian elimination. Also, we provide an upper bound on the complexity of chordal elimination.
\end{abstract}
% Note that the IEEE does not put floats in the very first column
% - or typically anywhere on the first page for that matter. Also,
% in-text middle ("here") positioning is typically not used, but it
% is allowed and encouraged for Computer Society conferences (but
% not Computer Society journals). Most IEEE journals/conferences use
% top floats exclusively. 
% Note that, LaTeX2e, unlike IEEE journals/conferences, places
% footnotes above bottom floats. This can be corrected via the
% \fnbelowfloat command of the stfloats package.
\section{Introduction}
Linear codes are widely used to protect information against losses in communication and storage systems. For a linear code with parity-check matrix $H\in\mathbb{F}_q^{m\times n}$, maximum-likelihood decoding over a $q$-ary erasure channel reduces to solving \begin{equation*} H_E x_E = -H_{E^c}x_{E^c}, \end{equation*} where $E$ is the set of erased coordinates and $H_E$ is the corresponding column submatrix~\cite{RichardsonUrbanke2008}. Therefore, decoding complexity is closely related to the cost of Gaussian elimination on submatrices of $H$. Gaussian elimination may introduce additional non-zero entries, known as fill-in (a zero entry in the matrix turning non-zero during the Gaussian elimination\cite{harb2025refill}), which increases computational and memory costs. Therefore, it is natural to design codes with specially structured parity-check matrices to reduce the computational complexity of Gaussian elimination. In the last decade, several codes have been proposed, such as LDPC-Band codes \cite{soro2009erasure}, pseudo-band form of the QC-LDPC codes \cite{cunche2010analysis} and polar codes with sparse parity-check representation \cite{urman2021efficient}.  In this paper, we propose a novel code construction with lower decoding complexity than conventional Gaussian-elimination decoding. Inspired by the perfect elimination property of chordal graphs, we design the parity-check matrix such that its primal graph is chordal. This structure allows Gaussian elimination to follow a perfect elimination ordering, thereby controlling fill-in and reducing the overall decoding complexity. The remainder of this paper is organized as follows. Section~II introduces the necessary theoretical concepts in graph theory and coding theory. Section~III presents the chordal code construction. In Section~IV, we compare the complexities of chordal elimination (Gaussian elimination following a perfect elimination ordering in a chordal graph) with Gaussian elimination through theoretical analysis as well as simulations. Section~V is the conclusion.
\addtolength{\textheight}{-0.05in}
\section{Preliminaries}\label{II}
In this section, we introduce some concepts in graph theory which will be used later in this paper. Also, our channel model is introduced in this section.
\subsection{Notation in Graph Theory}
\begin{definition}[Undirected graph]
    An undirected graph $G$ consists of a vertex set $V:=\{x_0,x_1\cdots,x_{n-1}\}$ and an edge set $E:=\{e_0,e_1,\cdots,e_{m-1}\}$. An edge can be represented by the two endpoints. For example, if there is an edge $e_i$ between two vertices $x_j$ and $x_k$, then edge $e_i$ can be represented as $(x_j,x_k)$. We say that $x_j$ and $x_k$ are \emph{adjacent}, or \emph{connected}, if there exists an edge whose two component vertices are $x_j$ and $x_k$. A sequence of edges $(x_1,x_2),(x_2,x_3),\cdots,(x_{i-1},x_i),(x_i,x_1)$ forms a \emph{cycle} of length $i$ for $i\geq 3$, represented by $x_1-x_2-\cdots-x_i-x_1$.
\end{definition}

\vspace{2mm}

\begin{definition}[Clique]
    Let $G=(V,E)$ be an undirected graph with vertex set $V=\{x_0,x_1,\cdots, x_{n-1}\}$ and edge set $E=\{e_0,e_1,\cdots,e_{m-1}\}$. A \emph{clique} $C$ is a subgraph of $G$ such that any two vertices in $C$ are adjacent, i.e.,  there is an edge between any two vertices in $C$.
\end{definition}

\vspace{2mm}

\begin{definition}[Induced Subgraph\cite{diestel2017extremal}]
    Let $G$ be an undirected graph with vertex set $V=\{x_0,x_1,\cdots, x_{n-1}\}$ and edge set $E=\{e_0,e_1,\cdots,e_{m-1}\}$. An \emph{induced subgraph} of $G$, denoted by $G[V']=(V',E')$, is a graph determined by a subset of vertices
$V'\subseteq V$, where \[
E'=\{\{u,v\}\in E : u,v\in V'\}.
\]
That is, $G'$ contains all vertices in $V'$ and all edges of
$G$ whose two endpoints both belong to $V'$. 
\end{definition}
\addtolength{\topmargin}{0.03in}
\vspace{2mm}
The chordal graph was first introduced in \cite{berge1967some}. It has several equivalent definitions. The following two definitions of the chordal graph will be used in the subsequent sections.
\begin{definition}[Chordal Graph\cite{cifuentes2016exploiting}]\label{chor}
    Let $G$ be an undirected graph with vertex set $V=\{x_0,x_1,\cdots, x_{n-1}\}$. An ordering of its vertices $x_0>x_1>\cdots >x_{n-1}$ is called a \emph{perfect elimination ordering} if for all $x_l\in V$, the induced subgraph $G[X_l]$ is a clique, where the set $X_l$ is defined as 
    \begin{equation}
        X_l:=\{x_l\}\cup\{x_m: x_m \mbox{ is adjacent to }x_l\mbox{ and } x_m<x_l\}.
    \end{equation}
    $G$ is called \emph{chordal} if it has a perfect elimination ordering.
\end{definition}

\vspace{2mm}
\addtolength{\textheight}{-0.05in}
Chordal graph can also be defined equivalently as follows:
\begin{definition}[Chordal Graph\cite{bondy2008graph}]
Let $G=(V,E)$ be an undirected graph with vertex set $V=\{x_0,x_1,\cdots, x_{n-1}\}$ and edge set $E=\{e_0,e_1,\cdots,e_{m-1}\}$. The graph $G$ is chordal if every cycle of length greater than three has a chord. 
    
\end{definition}

\vspace{2mm}
 Note that a chord is an edge that connects two vertices which are not adjacent in the cycle. For example, let $x_1-x_2-x_3-x_4-x_1$ be a cycle of length $4$. Since $x_1$ is only adjacent to $x_2$ and $x_4$, $x_1$ is not adjacent to $x_3$. So if edge $e=(x_1,x_3)$ exists, it is a chord which connects vertices $x_1$ and $x_3$.

%\subsection{Primal Graph}
\begin{definition}[Primal Graph]\label{primaldef}
    Let $H\in\mathbb{F}_q^{m\times n}$ be a matrix. A primal graph of $H$, denoted by $P(H)$, is an undirected simple graph defined as follows:
    \begin{itemize}
        \item The vertex set $V$ of $P(H)$ is \[V=\{1,2,\cdots, n\},\]
        where vertex $j$ corresponds to column $j$ in $H$.
        \item For each row $i\in\{1,2,\cdots, m\}$, define the vertex set \[V_i=\{j\in V| H_{i,j}\neq 0\}.\]
        The set $V_i$ induces a clique in $P(H)$, i.e., there is an edge between any two vertices in $V_i$.
    \end{itemize}
\end{definition}
\vspace{2mm}

Since $P(H)$ is an undirected simple graph, the cliques induced by different rows in $H$ can intersect or coincide.

Note that vertices $i,j$ in $P(H)$ are adjacent if and only if there exists at least one row $k$, such that $H_{k,i}\neq 0$ and $H_{k,j}\neq 0$.

\vspace{2mm}

As proved in Lemma 1 and Theorem 1 in \cite{rose1970triangulated}, there always exists a so-called \emph{monotone transitive ordering} in a chordal graph which ensures that the vertex elimination adds no edges. In Gaussian elimination,  this reduces the fill-in generated in the Gaussian elimination in the corresponding matrix. This monotone transitive ordering is in fact equivalent to the perfect elimination ordering (PEO)\cite{cifuentes2016exploiting}. Therefore, let $H$ be a matrix with $P(H)$ being chordal. Then we can reduce the complexity of the Gaussian elimination of $H$ by choosing the pivots under the perfect elimination ordering of $P(H)$. 

For a matrix under Gaussian elimination, a row is called \emph{active} if it has not been selected as a pivot row in any previous elimination step. We say that a column $x_l$ is \emph{processed} when it is examined for a pivot: if it contains a non-zero entry in an active row, a pivot row is selected and the column is eliminated from all other active rows. Otherwise, the column is skipped. 

\subsection{Linear Coding and Erasure Decoding}\label{dc}
Let $m<n$, and let $H\in\mathbb{F}_q^{m\times n}$ be a full-row-rank parity-check matrix. The associated linear code is
\begin{equation}
    \mathcal{C}
    :=
    \left\{
        \boldsymbol{x}\in\mathbb{F}_q^n :
        H\boldsymbol{x}=\boldsymbol{0}
    \right\},
\end{equation}
which is an $[n,n-m]_q$ linear code. Suppose that a codeword $\boldsymbol{x}\in\mathcal{C}$ is transmitted over a $q$-ary erasure channel.  Note that in the rest of the paper, we will use symbol $x$ to denote a vertex in the graph as well as a column vector of length $n$ with each symbol defined in base field $\mathbb{F}_q.$ Let $\mathcal{E}\subseteq[n]$ denote the set of erased coordinates, where $[n]:=\{1,\ldots,n\}$, and let $\mathcal{E}^{\mathrm{c}}=[n]\setminus\mathcal{E}$. For any $\mathcal{S}\subseteq[n]$, let $H_{\mathcal{S}}$ denote the submatrix of $H$ formed by the columns indexed by $\mathcal{S}$, and let
$\boldsymbol{x}_{\mathcal{S}}$ denote the corresponding subvector of $\boldsymbol{x}$. Since the erasure pattern and the unerased symbols are known at the decoder, every codeword consistent with the channel output must satisfy
\begin{equation}
    H_{\mathcal{E}}\boldsymbol{x}_{\mathcal{E}}
    =
    -H_{\mathcal{E}^{\mathrm{c}}}
    \boldsymbol{x}_{\mathcal{E}^{\mathrm{c}}}.
    \label{eq:erasure-decoding-system}
\end{equation}
For the erasure channel, all codewords consistent with the channel output have the same likelihood. Hence, maximum-likelihood decoding is equivalent to solving \eqref{eq:erasure-decoding-system}. The transmitted
codeword is uniquely recoverable if and only if
\begin{equation}
    \operatorname{rank}\!\left(H_{\mathcal{E}}\right)
    =
    |\mathcal{E}|,
\end{equation}
i.e., if and only if the columns of $H$ indexed by $\mathcal{E}$ are linearly independent.

\section{Chordal Code Scheme}
In this section, we will apply the chordal graph to the erasure channel.
As mentioned in Section \ref{dc}, if a code $\mathcal{C}$ is transmitted through an erasure channel, then its parity-check matrix $H$ will be replaced by a submatrix $H'$ by deleting some columns of $H$. 
\begin{lemma}\cite{diestel2017extremal}\label{l1}
    Let $G=(V,E)$ be a chordal graph and $V'\subset V.$ Then, $G[V']$ is chordal.
\end{lemma}
The above lemma shows that the induced subgraph of a chordal graph is also a chordal graph.
\begin{theorem}\label{subgraph}
    Let $H'$ be a column submatrix of $H\in\mathbb{F}_q^{m\times n}$, i.e., $H'$ is obtained by deleting some of the columns of $H$.If $P(H)$ is chordal, then $P(H')$ is chordal.
\end{theorem}

\begin{IEEEproof}
    Consider an $m\times n$ matrix $H$ with $c_1,c_2,\cdots,c_n$ being the columns of $H$. Let $H'$ be an $m\times l$ submatrix of $H$ with $c_{i_1},c_{i_2},\cdots,c_{i_l}$ being the columns of $H'$. Note that $l\leq n$ and $\{c_{i_1},c_{i_2},\cdots,c_{i_l}\}\subset\{c_1,c_2,\cdots , c_n\}$. So from Definition \ref{primaldef}, we obtain that $V_{P(H)}=\{1,2,\cdots,n\}$ and $V_{P(H')}=\{i_1,i_2,\cdots,i_l\}$. Obviously, $V_{P(H')}\subset V_{P(H)}$. So, to prove that $P(H')$ is an induced subgraph of $P(H)$, we only need to prove that for any two vertices in $V_{P(H')}$, they are adjacent in $P(H)$ if and only if they are adjacent in $P(H')$. 
    
    First, we show that for any two vertices in $V_{P(H')}$, if they are adjacent in $P(H')$, then they are adjacent in $P(H)$. Assume that $i_a,i_b\in V_{P(H')}$ are two arbitrary adjacent vertices in $P(H')$. According to Definition \ref{primaldef}, we obtain that column $c_{i_a}$ and column $c_{i_b}$ have at least one non-zero entry in the same row in $H'$. Since $H'$ is a column submatrix of $H$, it follows that column $c_{i_a}$ and column $c_{i_b}$ have at least one non-zero entry at the same row in $H$ too. So $i_a, i_b$ are adjacent in $P(H)$.
    
    Next, we show that for any two vertices in $V_{P(H')}$, if they are adjacent in $P(H)$, then they are adjacent in $P(H')$. Assume that $i_a,i_b\in V_{P(H')}$ are two arbitrary adjacent vertices in $P(H)$. By the same argument, we obtain that column $c_{i_a}$ and column $c_{i_b}$ have at least one non-zero entry at the same row in $H$.  Since $H'$ is a column submatrix of $H$ and $i_a,i_b\in V_{P(H')}$, it follows that column $c_{i_a}$ and column $c_{i_b}$ remain unchanged in $H'$. So $i_a, i_b$ are adjacent in $P(H')$. Therefore, $P(H')$ is an induced subgraph of $P(H)$. Together with Lemma \ref{l1}, $P(H')$ is chordal. Hence, the theorem is proved.
\end{IEEEproof}

\vspace{2mm}
\addtolength{\textheight}{-0.05in}
 From Subsection \ref{dc}, decoding a code $\mathcal{C}$ with parity-check matrix $H$ over an erasure channel reduces to solving the linear system (\ref{eq:erasure-decoding-system}) where $H_E$ is the column submatrix of $H$ indexed by the erased coordinates. By Theorem \ref{subgraph}, $P(H_E)$ is an induced subgraph of $P(H)$. Therefore, if $P(H)$ is chordal, then $P(H_E)$ is also chordal and admits a perfect elimination ordering (PEO). Consequently, Gaussian elimination can be performed according to this PEO, with lower complexity.

% \subsection{Fill-in}
% In the Gaussian elimination, a fill-in happens if an original zero entry in the matrix turns into a non-zero entry during the elimination. The complexity of the Gaussian elimination will increase if too many fill-ins happen. In this section,
\subsection{Code Construction}

\begin{definition}[Chordal code]
    Let $\mathcal{C}\subset\mathbb{F}_q^n$ be a code and $H$ be a parity-check matrix of $\mathcal{C}$. $\mathcal{C}$ is called \emph{chordal} if $\mathcal{C}$ is equipped with a parity-check matrix $H$ such that $P(H)$ is a chordal graph.
\end{definition}

The following theorem shows a construction of a chordal code $\mathcal{C}$ with $H$ being its parity-check matrix.

\begin{theorem}\label{chordalcode}
    Let $H_1$ be an $m_1\times n$ parity-check matrix of a random code in the systematic form, i.e., \[H_1:= \left[ P \mid I\ \right],\]
    where $P$ is an $m_1\times (n-m_1)$ sparse matrix.  Let $H_2$ be an $m_2\times n$ matrix with $m_2\leq n-m_1$ as follows:
    \[H_2:=\left[ L \mid \boldsymbol{0}\ \right],\]
    with $L$ being an $m_2\times (n-m_1)$ full rank matrix such that the primal graph of $L$ is a complete graph (clique). A simple way to ensure that $P(L)$ is complete is to include a row whose entries are all nonzero. $\boldsymbol{0}$ refers to an all-zero matrix of dimension $m_2\times m_1$. If $m_1+m_2<n$, then
   \[H:= \begin{bmatrix} H_2\\ H_1 \end{bmatrix}\]
   is a parity-check matrix of a chordal code $\mathcal{C}$ with length $n$ and dimension $n-m_1-m_2.$
\end{theorem}
\begin{IEEEproof}
    According to the dimensions of $H_1$ and $H_2$, it is obvious that $H$ is an $(m_1+m_2)\times n$ matrix. So it is a parity-check matrix of a code $\mathcal{C}$ with length $n$ and dimension $n-m_1-m_2$. So we only need to prove that $P(H)$ is a chordal graph.

    Let $x_1,x_2,\cdots,x_n$ be the $n$ columns in $H$ from left to right, and they correspond to $n$ vertices in $P(H)$. Recall 
    \[ H= \begin{bmatrix} L & \boldsymbol{0}\ \ \\ \underset{\smash{\raisebox{-1.8ex}{$ \underbrace{\hspace{1.5cm}}_{n-m_1\text{ columns}} $}}}{P} & I\ \ \end{bmatrix}. \]
    \vspace{3.6mm}
    
    First, we consider the first $n-m_1$ columns. Since $P(L)$ is a complete graph, vertices $x_1,x_2,\cdots,x_{n-m_1}$ form a clique according to definition \ref{primaldef}, i.e., any two vertices among $x_1,x_2,\cdots,x_{n-m_1}$ are adjacent. Then for the remaining $m_1$ columns, since $I$ is an $m_1\times m_1$ identity matrix and $\boldsymbol{0}$ is the all zero matrix, there is only one $1$ symbol in each column and all these $1$ symbols are in different rows. So any two vertices among $x_{n-m_1+1},x_{n-m_1+2},\cdots, x_n$ are not adjacent in $P(H)$. and each of them forms a clique with some of the vertices among $x_1,x_2,\cdots,x_{n-m_1}$. 
    
    Now we order the vertices as follows:
    \begin{equation}\label{order}
        x_{n-m_1+1}>x_{n-m_1+2}>\cdots x_n>x_1>x_2>\cdots>x_{n-m_1}.
    \end{equation}
    For any vertex among $x_{n-m_1+1},x_{n-m_1+2},\cdots, x_n$ in this ordering, it is only adjacent to some of the vertices among $x_1,x_2,\cdots,x_{n-m_1}$, which appear later in this ordering. These neighbors form a clique. On the other hand, since $x_1,x_2,\cdots,x_{n-m_1}$ form a clique, every vertex among $x_1,x_2,\cdots,x_{n-m_1}$ induces a clique together with all its adjacent vertices appearing later in the ordering. So, every vertex induces a clique with all its adjacent vertices appearing later in the ordering. So, (\ref{order}) is a perfect elimination ordering according to Definition \ref{chor}, and $P(H)$ is a chordal graph. Hence $\mathcal{C}$ is a chordal code and we prove the theorem.
\end{IEEEproof}

Note that it is not necessary for $L$ to have a row whose entries are all non-zero elements. In fact, if for all $x_i,x_j\in\{x_1,x_2,\cdots,x_{n-m_1}\}$, there exists at least one row $r$ in $H$ such that $H_{r,x_i}$ and $H_{r,x_j}$ are non-zero elements, then any two vertices among $x_1,x_2,\cdots,x_{n-m_1}$ are connected, implying that the first $n-m_{1}$ vertices form a clique. By the proof of Theorem \ref{chordalcode}, we obtain that $P(H)$ is a chordal graph. So, as long as $L$ ensures that for all $x_i,x_j\in\{x_1,x_2,\cdots,x_{n-m_1}\}$, there exists at least one row $r$ in $H$ such that $H_{r,x_i}$ and $H_{r,x_j}$ are non-zero elements, $P(H)$ is a chordal graph and $\mathcal{C}$ is a chordal code.

\subsection{Toy example}
Here we give a toy example of a chordal code $\mathcal{C}$ with length $n=7$ and $k=4$. The parity-check matrix $H$ is as follows:
\begin{equation}
    H=\begin{bmatrix} 1,1,1,1,1,0,0\\0,1,0,0,1,1,0 \\ 
    1,0,1,0,0,0,1

    \end{bmatrix}
\end{equation}
with $L=[1,1,1,1,1]$ and \[P=\begin{bmatrix}
    0,1,0,0,1\\1,0,1,0,0
\end{bmatrix}.\]

\begin{figure}
    \centering
    \includegraphics[scale=0.45]{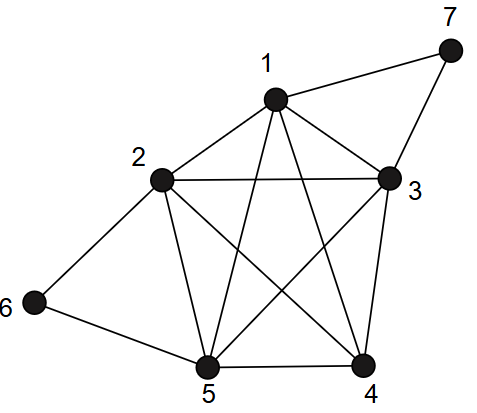}
    \caption{Primal graph of the parity-check matrix $H$ of $\mathcal{C}$}\label{fig1}
    
\end{figure}
\addtolength{\textheight}{-0.05in}

Figure \ref{fig1} shows the primal graph of $H$. Note that vertex $i$ corresponds to the $i^{th}$ column of $H$. Obviously, $P(H)$ is a chordal graph because for every cycle of length greater than 3, there always exists at least one chord. For example, cycle $2-6-5-4-2$ is a length-$4$ cycle and edge $(2,5)$ is the chord. The cycle $1-2-5-3-7-1$ is a length-$5$ cycle and it has three chords, which are $(1,5), (2,3)$ and $ (1,3)$.
\section{Chordal Elimination Complexity }
  In this section, the complexity of chordal elimination of a chordal code in $\mathbb{F}_2^n$ constructed by Theorem \ref{chordalcode} will be analyzed and an upper bound will be given. Moreover, the complexities of chordal elimination and Gaussian elimination on a chordal code in $\mathbb{F}_2^n$ will be compared through simulations.

\subsection{Complexity of the Chordal Elimination}
Let $H\in\mathbb{F}_2^{m\times n}$ be an $m\times n$ matrix of rank $r$ with the largest clique having $\kappa$ vertices. In this subsection, we use the number of field additions taken by the Gaussian elimination as the complexity. This is because in binary fields, it is not necessary to calculate the inverse of $1$. Using Lemma \ref{4.1} and Theorem \ref{4.2}, we will show that the complexity of the chordal elimination is upper bounded by $mr\kappa$.
\vspace{2mm}

\begin{lemma}\label{4.1}
    Let $H\in\mathbb{F}_2^{m\times n}$ be a matrix with $G=P(H)$ being chordal. $x_1>x_2>\cdots> x_n$ is the PEO. For each $t\in\{1,\ldots,n\}$, let $H^{(t)}$ be the matrix formed by the active rows and the unprocessed columns immediately before column $x_t$ is processed. Let \[ V_t:=\{x_t,x_{t+1},\ldots,x_n\}. \] Then \[ P(H^{(t)})\subseteq G[V_t]. \] In other words, Gaussian elimination according to the PEO creates no new adjacency outside the corresponding induced subgraph of the original primal graph.
\end{lemma}
\begin{IEEEproof} We prove the statement by induction on $t$. For $t=1$, no row operation has been performed, and all columns and rows are active. So, \[ H^{(1)}=H \] and \[V_1=\{x_1,x_2,\cdots,x_n\}.\] Therefore \[ P(H^{(1)})=P(H)=G=G[V_1]. \] Hence, the claim holds for $t=1$. 

Now assume \[ P(H^{(t)})\subseteq G[V_t] \] holds before column $x_t$ is processed. If column $x_t$ has zero entries in all active rows, no row operation is performed. Therefore, by removing $x_t$ from the set of unprocessed columns, we have\[P(H^{(t+1)})\subseteq P(H^{(t)})[V_{t+1}]\subseteq G[V_{t+1}].\]  On the other hand, if column $x_t$ has a non-zero entry in at least one active row, let $R_t$ be the pivot row having a non-zero entry in column $x_t$.  Define \[ C_t:=\{x_t\}\cup N_G(x_t), \] where \[ N_G(x_t):= \{x_j\in V_{t+1}: x_j \mbox{ is adjacent to } x_t\} \] is the set of unprocessed neighbors of $x_t$ in $G$. So $C_t$ is exactly the set of vertex $x_t$ and all its unprocessed neighbors. 
Since $x_1>x_2>\cdots>x_n$ is a PEO, the set $C_t$ induces a clique in $G$ by Definition \ref{chor}. Now consider every other non-pivot active row $R$ having a non-zero entry in column $x_t$, they will be replaced by $R+R_t$ when processing $x_t$ so that the non-zero entry turns to 0. Note that \[\mbox{supp}(R_t)\subseteq C_t\] by the induction hypothesis. For all $R$ which is still active when processing $x_t$, any non-zero column $x_j$ in $R$ is a neighbor of $x_t$ and $x_j\in V_{t+1}$. So $x_j\in C_t$. So, \[\mbox{supp}(R)\subseteq C_t.\] Therefore, when $R$ is replaced by $R+R_t$, the fill-in will only happen in the columns in $C_t$, i.e., \[\mbox{supp}(R+R_t)\subseteq \mbox{supp}(R)\cup\mbox{supp}(R_t)\subseteq C_t.\] So the row operations when processing $x_t$ will not create new edges in $P(H^{(t)})$. So,  after column $x_t$ has been eliminated, its pivot row is removed from the active rows and $x_t$ is removed from the set of unprocessed columns, implying that $P(H^{(t+1)})$ is a subgraph of $G[V_t]$. So,\[P(H^{(t+1)})\subseteq G[V_{t+1}].\]
Hence, by induction, the lemma is proved.
\end{IEEEproof}

\vspace{2mm}

\begin{theorem}\label{4.2}
    Let $H$ be a matrix in $\mathbb{F}_2^{m\times n}$ whose primal graph $G=P(H)$ is chordal. Let $\kappa$ be the number of vertices of the largest clique in $G$. Then the complexity of the chordal elimination (Gaussian elimination following the PEO) is less than  $m\cdot r\cdot \kappa,$ with $r=\mbox{rank}(H)$. 
\end{theorem}
\begin{IEEEproof}
Let $H$ be an $m\times n $ matrix whose primal graph $G$ is chordal. Let $x_1,x_2,\cdots,x_n$ be the vertices of $G$. Without loss of generality, we assume that $x_1>x_2>\cdots>x_n$ is the perfect elimination ordering. The chordal elimination, processed by the decoder, consists of $n$ elimination steps, with $r=\mbox{rank}(H)$ of them being \emph{successful elimination steps}. At each elimination step $t$, column $x_{t}$ is called the pivot column. If the pivot column $x_{t}$ has no non-zero entry, then no pivot row is selected and the column is skipped. Otherwise, the decoder chooses a unique pivot row and adds the pivot row to every other row having a non-zero entry in $x_{t}$. Therefore, the pivot column $x_{t}$ has exactly one non-zero entry in the pivot row and has zero entries in other rows. Such an elimination step is then called \emph{successful}. Let $x_{i_1},x_{i_2},\cdots, x_{i_r}$ be the pivot columns in successful elimination steps with $i_1<i_2<\cdots<i_r$. Note that field additions are performed in successful elimination steps.

 Now we prove  the theorem by induction. First, consider the pivot column $x_{i_1}$ in the first successful elimination step. Let $R_1$ be the selected pivot row. Since $R_1$ has a non-zero entry in column $x_{i_1}$, every other unprocessed column whose entry in $R_1$ is non-zero is adjacent to $x_{i_1}$ in the primal graph of the current active matrix $H^{(i_1)}=H$. By Lemma \ref{4.1}, $\operatorname{supp}(R_1)\subseteq C_{i_1}. $  So, adding $R_1$ to another active row requires at most \[ |\operatorname{supp}(R_1)|\leq |C_{i_1}|\leq\kappa \] field additions. Since there are at most $m-1$ non-pivot active rows, the first successful elimination step requires at most \[ (m-1)|C_{i_1}|<m\kappa \] field additions.
 
 Next, assume that the first $l-1$ successful elimination steps require fewer than $ m(l-1)\kappa $ field additions for $l\geq 2$. Consider the $l$-th pivot column $x_{i_l}$, and let $R_l$ be its pivot row. Since all previously processed columns have zero entries in all the remaining active rows, $R_l$ has non-zero entries only in the set of currently unprocessed columns. By Lemma \ref{4.1}, $\operatorname{supp}(R_l)\subseteq C_{i_l}.$ Note that $C_{i_l}$ forms a clique in $G$. Then it follows that \[ |\operatorname{supp}(R_l)| \leq |C_{i_l}| \leq \kappa. \] Since the pivot row $R_l$ is added only to the remaining active non-pivot rows having a non-zero entry in column $x_{i_l}$, there are at most $m-1$ such rows, requiring at most $ |\operatorname{supp}(R_l)|\leq \kappa $ field additions for each row. Thus, the $l$-th successful elimination step requires at most \[ (m-1)|C_{i_l}|<m\kappa \] field additions. Combining with the induction hypothesis, the first $l$ successful elimination steps require fewer than $ ml\kappa $ field additions. Therefore, by induction, all $r$ successful elimination steps require fewer than $mr\kappa$ field additions. Hence, the theorem is proved.
\end{IEEEproof}

Note that under the design of Theorem \ref{chordalcode}, the largest clique is usually the clique created by the rows in $L$, i.e., $\kappa=n-m_1$. On the other hand, the complexity of the Gaussian elimination is roughly upper bounded by $mrn$. Hence, reduction of the complexity comes from replacing code length $n$ by the maximum clique size $\kappa$.

\subsection{Chordal Elimination Performance}
In this subsection, we compare the complexity of chordal elimination with the Gaussian elimination through $L=100$ simulations.
\begin{table}[t] \vspace*{0.03in}\caption{Simulations for the Chordal Code With $k=30$} \label{tab:simulation-k20} \centering \footnotesize \setlength{\tabcolsep}{4.5pt} \renewcommand{\arraystretch}{1.15} \begin{tabular}{lcc}  \hline Operation & Gaussian & Chordal\\ \hline Field additions (total complexity) & 36596 & 29462\\ Average complexity per trial & 365.96 & 294.62\\ \hline Complexity reduction & \multicolumn{2}{c}{19.494\%}\\ \hline \hline \end{tabular} \end{table}
 % \begin{table}[t] \caption{Simulation Parameters for the Chordal Codes With $k=20$} \label{para2} \centering \footnotesize \setlength{\tabcolsep}{4.5pt} \renewcommand{\arraystretch}{1.15} \begin{tabular}{lcc} \hline \hline \multicolumn{3}{c}{\textbf{Simulation Parameters}}\\ \hline Number of simulations, $L$ & \multicolumn{2}{c}{100}\\ Code length, $n$ & \multicolumn{2}{c}{50}\\ Code dimension parameter, $k$ & \multicolumn{2}{c}{20}\\ Field size, $q$ & \multicolumn{2}{c}{2}\\ Column erasure probability & \multicolumn{2}{c}{0.3000}\\ Average number of added rows & \multicolumn{2}{c}{5.000}\\ Average number of deleted columns & \multicolumn{2}{c}{15.240}\\ Unique decoding instances & \multicolumn{2}{c}{100}\\ Non-unique consistent instances & \multicolumn{2}{c}{0}\\ \hline \hline \end{tabular} \end{table}
\begin{table}[t] \caption{Simulations for the Chordal Codes With $k=20$} \label{tab:simulation-k30} \centering \footnotesize \setlength{\tabcolsep}{4.5pt} \renewcommand{\arraystretch}{1.15} \begin{tabular}{lcc}   \hline Operation & Gaussian & Chordal\\ \hline Field additions (total complexity) & 51,692 & 40,003 \\ Average complexity per trial & 516.92 & 400.03\\ \hline Complexity reduction & \multicolumn{2}{c}{22.613\%}\\ \hline \hline \end{tabular} \end{table} In each simulation, we generate a random parity-check matrix $H$ in $\mathbb{F}^{(n-k+m)\times n}_q$, with $n=50,k=30,m=5$ and $q=2.$ $n$ is the code length, $k-m$ is the dimension of the code, i.e., the code has $q^{k-m}$ codewords. The code rate of the chordal code is $\frac{(30-5)}{50}=0.5$ When passing an erasure channel with erasure probability $0.3$, i.e., every symbol of the transmitted codeword $x$ has a $30\%$ probability to be erased. Among all $100$ simulations, $96$ of them decode successfully and $4$ of them have a decoding failure. 

% \begin{figure}
%     \centering
%     \includegraphics[scale=0.7]{simulations.png}
%     \caption{Comparison of chordal elimination and Gaussian elimination}\label{fig2}
 
% \end{figure}
Table \ref{tab:simulation-k20} shows the complexities of Gaussian elimination and the chordal elimination in the case that $k=30$.
Note that in the simulations of the case that $k=30$, the successful rate is $96\%$. If a higher successful rate is required, we could lower the code rate to improve the successful rate. Table \ref{tab:simulation-k30} shows another result about the simulations of a lower code rate. We take 100 simulations again and the only difference is that we lower the code rate from $0.5$ to $\frac{20-5}{50}=0.3$ by changing $k$ from $30 $ to $20$. decoding succeeds in all 100 simulations.

% \begin{figure}
%     \centering
%     \includegraphics[scale=0.72]{simulations1.png}
%     \caption{Comparison of chordal elimination and Gaussian elimination}\label{fig3}
 
% \end{figure}
\section{Conclusion} In this paper, we introduced chordal graphs and defined chordal codes. We showed that the chordality of a parity-check matrix's primal graph is preserved under column erasure, so that the decoding using chordal elimination can be processed adding less fill-in by following the perfect elimination ordering. We also presented an explicit construction of chordal codes and proved that the chordal elimination has complexity at most $mr\kappa$, where $\kappa$ is the maximum clique size of the primal graph and $r$ is the rank of the parity-check matrix. Future work may investigate code constructions that achieve better trade-offs among code rate, erasure-correction capability, and decoding complexity.

\clearpage
\section*{Acknowledgments}
{\textit{AI statement:} The authors used ChatGPT for assistance with checking the correctness of proofs, and language polishing in the main text. The authors are responsible for the mathematical claims and the final manuscript.}

\bibliographystyle{IEEEtran}
\bibliography{ref.bib}

@book{RichardsonUrbanke2008,
  author    = {Thomas Richardson and R{\"u}diger Urbanke},
  title     = {Modern Coding Theory},
  publisher = {Cambridge University Press},
  year      = {2008}
}

@article{cifuentes2016exploiting,
  title={Exploiting chordal structure in polynomial ideals: A grobner bases approach},
  author={Cifuentes, Diego and Parrilo, Pablo A},
  journal={SIAM Journal on Discrete Mathematics},
  volume={30},
  number={3},
  pages={1534--1570},
  year={2016},
  publisher={SIAM}
}

@article{rose1970triangulated,
  title={Triangulated graphs and the elimination process},
  author={Rose, Donald J},
  journal={Journal of Mathematical Analysis and Applications},
  volume={32},
  number={3},
  pages={597--609},
  year={1970},
  publisher={Elsevier}
}

@article{berge1967some,
  title={Some classes of perfect graphs, Graph Theory and Theoretical Physics},
  author={Berge, Claude and Harary, F},
  journal={Academic Press, London},
  volume={155},
  pages={165},
  year={1967}
}

@incollection{diestel2017extremal,
  title={Extremal graph theory},
  author={Diestel, Reinhard},
  booktitle={Graph theory},
  year={2017},
  publisher={Springer}
}

@article{bondy2008graph,
  title={Graph theory, volume 244 of Graduate Texts in Mathematics. ISBN: 978-1849966900 Springer},
  author={BONDY, JA and MURTY, USR},
  journal={New York},
  year={2008}
}

@article{harb2025refill,
  title={Refill: Reinforcement learning for fill-in minimization},
  author={Harb, Elfarouk and Lam, Ho Shan},
  journal={arXiv preprint arXiv:2501.16130},
  year={2025}
}

@inproceedings{soro2009erasure,
  title={Erasure codes with a banded structure for hybrid iterative-ML decoding},
  author={Soro, Alexandre and Cunche, Mathieu and Lacan, J{\'e}r{\^o}me and Roca, Vincent},
  booktitle={GLOBECOM 2009-2009 IEEE Global Telecommunications Conference},
  pages={1--6},
  year={2009},
  organization={IEEE}
}

@inproceedings{urman2021efficient,
  title={Efficient maximum likelihood decoding of polar codes over the binary erasure channel},
  author={Urman, Yonatan and Burshtein, David},
  booktitle={2021 IEEE International Symposium on Information Theory (ISIT)},
  pages={1600--1605},
  year={2021},
  organization={IEEE}
}

@inproceedings{cunche2010analysis,
  title={Analysis of quasi-cyclic LDPC codes under ML decoding over the erasure channel},
  author={Cunche, Mathieu and Savin, Valentin and Roca, Vincent},
  booktitle={2010 International Symposium On Information Theory \& Its Applications},
  pages={861--866},
  year={2010},
  organization={IEEE}
}

% or


\end{document}